\documentclass{article}%
\usepackage[margin=1in]{geometry}
\usepackage{cite}
\usepackage{graphicx}
\usepackage{titlesec}
\usepackage[T1]{fontenc}
\usepackage{amsmath}
\usepackage{amsfonts}
\usepackage{amssymb}
\usepackage[colorlinks,linkcolor=blue]{hyperref}
\usepackage{amsthm}
\usepackage{authblk}%
\providecommand{\U}[1]{\protect\rule{.1in}{.1in}}
\titleformat{\section}
{\large\normalfont\scshape}{\thesection}{1em}{}
\renewcommand{\thesection}{\Roman{section}}
\theoremstyle{definition}
\newtheorem{theorem}{Theorem}

\newtheorem{corollary}{Corollary}

\newtheorem{definition}{Definition}

\newtheorem{lemma}{Lemma}
\newtheorem{notation}{Notation}

\newtheorem{remark}{Remark}

\begin{document}

\title{Intrinsic Construction of Lyapunov Functions on Riemannian Manifold}
\author{Dongjun Wu\footnote{The author is with Center for Control Theory and Guidance Technology,
Harbin Institute of Technology, Harbin, 150001, China and Laboratoire des
Signaux et Syst\`emes, Universit\'e Paris Saclay, 91192 Gif-sur-Yvette, France;
\textit{Email:dongjun.wu@l2s.centralesupelec.fr}}}

\date{}
\maketitle

\begin{abstract}
For systems evolving on a Riemannian manifold, we propose converse Lyapunov
theorems for asymptotic and exponential stability. The novelty of the proposed
approach is that is does not rely on local Euclidean coordinate, and is thus
valid on a wider domain than the immediate vicinity of the considered
equilibrium. We also show how the constructions can be useful for robustness analysis.

\end{abstract}

\section{Introduction}

Since the foundation of modern control theory, Lyapunov's second method has
been among the most important tools for the investigations of stability of
control systems. The success of this method owes partially to the fact that it
is applicable for a wide class of dynamical systems, including continuous and
discrete time systems, stochastic systems, systems evolving on manifolds,
distributed parameter systems, delay systems, etc. For each scenario, the key
procedure of the Lyapunov's second method is the seeking of a Lyapunov
function. Before searching for such a function, we may ask whether it exists
and, if yes, how to construct it. These joint are known as the \emph{converse
Lyapunov problem}. A thorough review of the problem regarding continuous time
systems evolving on Euclidean space can be found in
\cite{yoshizawa1966stability}, \cite{teel2000smooth}. For more recent
developments, we refer the reader to a survey paper
\cite{kellett2015classical} and the references therein. The importance of the
converse problem lies not only in its theoretic interests but also its help to
the analysis of some other control problems. Typical applications of converse
theorem can be found in robust stability \cite{lin1996smooth} and
input-to-state stability analysis \cite{khalil2002nonlinear}.

In this paper, we study the converse Lyapunov problem for systems evolving on
Riemannian manifolds, for which not much attention has been drawn upon. It is
true that most control problems are studied in Euclidean spaces; the reason
seems to be that in most cases, stability is discussed in local sense, so the
manifold can be seen as Euclidean in a local coordinate. On the one hand, this
does simplify the problem, on the other hand, this may at the same time
preclude the utilization of the underlying rich structure of the manifold, for
example, to design coordinate-free control laws. In contrast to this
situation, a whole field of research is devoted to geometric tool for system
analysis and design, see for example, \cite{agrachev2013control},
\cite{jurdjevic1997geometric}, \cite{bullo2019geometric}. By ``geometric''\ we
mean coordinate-free or intrinsic objects. However, in terms of converse
Lyapunov problems, theories in parallel with those in the classical Euclidean
space still need to be established.

We list some previous works related to the problem that is going to be treated
in this paper, from which we will see that several issues need to be settled.
In 2010, F. Pait et al. \cite{pait2010some} suggested to use the quadratic
distance defined therein as a Lyapunov function. However, no explicit form of
such a function was not derived. In \cite{nakamura2013asymptotic}, the authors
constructed semiconcave control Lyapunov functions for general manifolds, by
relying on some constructions in local coordinates instead of the intrinsic
metric on the manifold. Besides, the constructed Lyapunov function is not
smooth. Later, F. Taringoo et al. systematically studied the problem on
Riemannian manifold \cite{taringoo2013local}. For example, locally
asymptotically stable and exponentially stable systems on Riemannian manifolds
were proved to admit certain Lyapunov functions. The theorems are proved by
``pulling back'' the Lyapunov function constructed in Euclidean space to the
manifold. The problem of this procedure is that it relies on the local
coordinate of the manifold hence the result is local, and further analysis on
whether the construction can be extended to the region of attraction are needed.

This paper aims at solving the issues mentioned in the previous paragraph in
this paper. Coordinate-free constructions will be given, and in the proof we
do not resort to the existing theories in Euclidean space. The Lyapunov
functions to be constructed will be seen to hold exactly the same meaning as
that in Euclidean space. The two crucial tools in the proof of our result is
the modified version of Lipschitz continuity and the first variation formula
of arc length.

The paper is organized as follows. In Section \ref{sec: pre}, we review some
basic definitions from Riemannian geometry and stability notions on Riemannian
manifolds. In Section \ref{sec: conv thm}, the main theorems are proved,
namely the construction of Lyapunov function for GAS and GES systems in a
coordinate-free manner. Section \ref{sec: dis and ext} includes some further
discussions and extensions of the main theorem, particularly, an application
to input-to-state stability is given.

\begin{notation}
Throughout this paper, we adopt the following notations. $\mathcal{X}$: the
riemannian manifold; $\Gamma(\mathcal{X})$: the set of smooth vector fields on
$\mathcal{X}$; $T_{x}\mathcal{X}$: the tangent space at $x$; $\left\langle
v_{x},u_{x}\right\rangle $: the Riemannian product of $v_{x},u_{x}\in
T_{x}\mathcal{X}$; $\nabla$: Levi-Civita connection; $\text{D}/\text{d}t$: the
Covariant derivative; $\ell(c)$: the length of the curve $c$; $d(x,y)$: the
Riemannian distance between $x$ and $y$; $L_{f}V(x)$: the Lie derivative of
$V(x)$ along the flow generated by $f(x)$; $\mathcal{L}_{f} V(t,x)$: the timed
Lie derivative of $V(t,x)$ along the flow of $f(t,x)$; $P_{p}^{q}$: the
parallel transport from $p$ to $q$; $\phi_{\ast}$: the push forward of a
diffeomorphism $\phi:\mathcal{M}\rightarrow\mathcal{N}$; $\phi^{\ast}$: the
pull back of a smooth map $\phi:\mathcal{M}\rightarrow\mathcal{N}$;
$|\cdot|_{\infty}$: the infinity norm; $\mathbb{R}_{+}$: the set of non
negative numbers; $B_{x}^{c}$: the open ball with radius $c$ centered at $x$;
$\phi(t;t_{0};x_{0})$: the flow of a system with initial condition $(t_{0},
x_{0})$.
\end{notation}

\section{Preliminaries\label{sec: pre}}

\subsection{Riemannian manifolds}

In this section, we summarize a few important tools that we will use
repeatedly from Riemannian geometry and the stability notions on Riemannian
manifold. A standard treatment of Riemannian geometry can be found in
\cite{carmo1992riemannian}.

On a Riemannian manifold, connections can be defined, among which there exists
an important one called the Levi-Civita connection.

\begin{theorem}
[Levi-Civita]Given a Riemannian manifold $\mathcal{X}$, there exists a unique
affine connection on $\mathcal{X}$, such that

\begin{enumerate}
\item[a.] $\nabla$ is symmetric, namely, $\nabla_{X}Y-\nabla_{Y}X=[X,Y]$ for
all $X,Y\in\Gamma(\mathcal{X})$.

\item[b.] $\nabla$ is compatible with the Riemannian metric, namely,
\begin{equation}
X\left\langle Y,Z\right\rangle =\left\langle \nabla_{X}Y,Z\right\rangle
+\left\langle Y,\nabla_{X}Z\right\rangle ,\ X,Y,Z\in\Gamma(\mathcal{X})
\label{comp Levi}%
\end{equation}

\end{enumerate}
\end{theorem}

From (\ref{comp Levi}) we get the following important equality:$\frac
{\text{d}}{\text{d}t}\left\langle V,W\right\rangle =\left\langle
\frac{\text{D}}{\text{d}t}V,W\right\rangle +\left\langle V,\frac{\text{D}%
}{\text{d}t}W\right\rangle $ where $V,W$ are vector fields along a given curve
parametrized by $t$. Throughout this paper, we use $\nabla$ to denote the
underlying Levi-Civita connection for the Riemannian manifold $\mathcal{X}$.

The Levi-Civita connection can be used to describe the geodesic curve, which
will serve as the crucial object in our paper. Loosely speaking, given two
points $x,y\in\mathcal{X}$, if there exists a $\mathcal{C}^{2}$ curve $\gamma$
joining $x$ to $y$ such that the length of $\gamma$ is minimized under small
variations, then $\gamma$ should verify $\nabla_{\gamma^{\prime}(s)}%
\gamma^{\prime}(s)=0$ for all $s\in I$. In this case $\gamma$ is called a
geodesic. Given $(x,v)\in T_{x}\mathcal{X}$, there exists (locally) a unique
geodesic $\gamma:(-\delta,\delta)$ such that $\gamma^{\prime}(0)=v$, where
$\delta$ may depend on $v$.

Once the Riemannian metric is given, the Riemannian distance $d(x,y)$ between
two points $x$ and $y$ can be defined which makes $\mathcal{X}$ into a metric
space. According to Hopf-Rinow theorem, this metric space is complete for
complete Riemannian manifold. Throughout this paper, we always assume the
Riemannian manifold to be complete.

We also mention an important theorem which will be used repeatedly throughout
this paper.

\begin{theorem}
[First variation of arc length]Let $\gamma:\left[  a,b\right]  \rightarrow
\mathcal{X}$ be a differentiable curve and $F$ a variation of $\gamma$. Denote
$\ell(c)$ the length of curve $c$. Then%
\begin{equation}
\left.  \frac{d}{dt}\right\vert _{t=0}\ell(F(t,\cdot))=\frac{1}{|\partial
F/\partial s|}\left[  \left\langle \frac{\partial F}{\partial s}%
,\frac{\partial F}{\partial t}\right\rangle _{s=a}^{s=b}-\int_{a}%
^{b}\left\langle \frac{\partial F}{\partial t},\nabla_{\partial/\partial
s}\frac{\partial F}{\partial s}\right\rangle \right]  _{t=0}, \label{1st Var}%
\end{equation}
which is called the first variation formula of arc length. When $\gamma$ is a
geodesic curve parametrized proportional to arc length, then%
\[
\left.  \frac{d}{dt}\right\vert _{t=0}\ell(F(t,\cdot))=\left\langle
\frac{\partial F}{\partial s},\frac{\partial F}{\partial t}\right\rangle
_{s=a}^{s=b}.
\]

\end{theorem}

\subsection{Stability notions on Riemannian manifolds}

Consider the following dynamical system evolving on $\mathcal{X}$:%
\begin{equation}
\dot{x}=f(t,x) \label{Sys}%
\end{equation}
where $f(t,x)$ is a time varying $\mathcal{C}^{1}$ vector field on
$\mathcal{X}$. An equilibrium point $x_{0}$ is such that $\phi(t;t_{0}%
,x_{0})=x_{0}$ for all $t\geq t_{0}$, where $\phi$ is the flow of $f$. Various
stability concepts can be given similarly as in Euclidean space by replacing
the norm by the Riemannian distance $d$ on $\mathcal{X}$.

\begin{definition}
\label{def: stability}An equilibrium $x_{\ast}$ for $\mathcal{X}$ is

\begin{enumerate}
\item (locally) uniformly\ stable (US) if there exists a class $\mathcal{K}$
function $\alpha$ and a positive constant $c$, independent of $t_{0}$, such
that%
\[
d(\phi(t;t_{0},x_{0}),x_{\ast})\leq\alpha(d(x_{0},x_{\ast})),\ \forall t\geq
t_{0}\geq0,\ \forall x_{0}\in B_{x_{\ast}}^{c};
\]

\item (locally)\ uniformly asymptotically stable (UAS) if there exists a class
$\mathcal{KL}$ function $\beta$ and a positive constant $c$, independent of
$t_{0}$, such that%
\begin{equation}
d(\phi(t;t_{0},x_{0}),x_{\ast})\leq\beta(d(x_{0},x_{\ast}),t-t_{0}),\ \forall
t\geq t_{0}\geq0,\ \forall x_{0}\in B_{x_{\ast}}^{c}. \label{eq: def: LAS}%
\end{equation}

\item (locally) exponentially stable (LES) if there exists three positive
constants $K$, $\lambda$ and $c$ such that%
\begin{equation}
d(\phi(t;t_{0},x_{0}),x_{\ast})\leq Ke^{-\lambda(t-t_{0})}d(x_{0},x_{\ast
}),\ \forall t\geq t_{0}\geq0,\ \forall x_{0}\in B_{x_{\ast}}^{c}.
\label{eq: def: LES}%
\end{equation}

\item uniformly globally asymptotically stable (UGAS) if (\ref{eq: def: LES}%
)is satisfied for any $x_{0}$; uniformly exponentially asymptotically (UGES)
if (\ref{eq: def: LES}) is satisfied for all $x_{0}$.
\end{enumerate}
\end{definition}

\begin{remark}
In \cite{bullo2019geometric}, \cite{taringoo2013local}, stability definitions
are given by the $\varepsilon$-$\delta$ language. However, it's not hard to
show that the two ways are equivalent, see for example
\cite{taringoo2013local}. Similar to the Euclidean case, the comparison
function will simplify the stability analysis.
\end{remark}

A Lyapunov candidate $V:\mathbb{R}_{+}\times U\rightarrow\mathbb{R}_{+}$ is a
locally positive definite function about the equilibrium $x_{\ast}$, where $U$
is an open neighborhood of $x_{\ast}$, namely, $V(t,x)\geq0$ for all
$t\in\mathbb{R}_{+}$ and $x\in U$ and $V(t,x)=0$ if and only if $x=x_{\ast}$.

On manifold, the partial derivative of a function is normally not a
coordinate-free notion. And in control systems, we will deal with time varying
vector fields, so we need the concept of \emph{timed Lie derivative}.

\begin{definition}
\label{def: tLie} [Timed Lie derivative]Consider the time invariant system%
\begin{align}
\frac{dx(t)}{dt}  &  =f(s(t),x(t))\nonumber\\
\frac{ds(t)}{dt}  &  =1, \label{TV to TIV}%
\end{align}
with initial condition%
\[
x(t_{0})=x_{0},\ s(t_{0})=t_{0}.
\]
The Lie derivative of $V$ with respect to (\ref{Sys}) is defined as the Lie
derivative of $V$ with respect to (\ref{TV to TIV}), and is denoted
$\mathcal{L}_{f}V$ more precisely,%
\begin{equation}
\mathcal{L}_{f}V(t,x)=:L_{\tilde{f}}V(t,x) \label{Lie der}%
\end{equation}
where $f$ is the augmented vector field $(f(s,x),\ 1)$.
\end{definition}

\begin{remark}
$L_{\tilde{f}}V$ is well defined since it's the usual Lie derivative of a
time-invariant function with respect to a time-invariant vector field. Since
the flow of $\tilde{f}$ is $(t,\phi_{f}(t;t_{0},x_{0}))$. In coordinates, at
point $(t_{0},x_{0})$, it reads%
\begin{align*}
L_{\tilde{f}}V(t_{0},x_{0})  &  =\lim_{t\rightarrow t_{0}}\frac{V(t,\phi
_{f}(t;t_{0},x_{0}))-V(t_{0},x_{0})}{t-t_{0}}\\
&  =\frac{\partial V}{\partial t}(t_{0},x_{0})+\frac{\partial V}{\partial
x}(t_{0},x_{0})f(t_{0},x_{0}),
\end{align*}
which coincides with the time derivative of $V$ along (\ref{Sys}).
\end{remark}

We are now in position to give the Lyapunov stability theorem on Riemannian manifolds.

\begin{theorem}
Let $x_{\ast}$ be an equilibrium point of the system (\ref{Sys}) and $D$ be an
open connected neighborhood of $x_{\ast}$. Let $V:\mathbb{R}_{+}\times
D\rightarrow\mathbb{R}_{+}$ be a Lyapunov candidate such that%
\begin{equation}
W_{1}(d(x,x_{\ast}))\leq V(t,x)\leq W_{2}(d(x,x_{\ast})),\ \forall
t\geq0,\ x\in D, \label{bd of LF}%
\end{equation}
then $x_{\ast}$ is uniformly stable if
\[
\mathcal{L}_{f}V(t,x)\leq0;
\]
it is uniformly asymptotically stable if%
\begin{equation}
\mathcal{L}_{f}V(t,x)\leq-W_{3}(d(x,x_{\ast})), \label{bd of dV}%
\end{equation}
where $W_{i}$ are class $\mathcal{K}$ functions. If $W_{i}(r)=c_{i}r^{p}$,
where $c_{i}>0$, for $i=1,2,3$ and $p>0$, then (\ref{bd of LF}) and
\ref{bd of dV}) together imply exponentially stability.
\end{theorem}

This theorem can be proved by repeating the procedures used in Euclidean space
by noticing that%
\[
\frac{\text{d}}{\text{d}t}V(t,\phi(t;t_{0},x_{0}))=\mathcal{L}_{f}V(t,x).
\]

\section{Converse theorem on Riemannian manifold\label{sec: conv thm}}

Though the converse theorem can be done locally, in order to simplify the
analysis and computation, and to streamline our idea, in the sequel, we will
assume global stability. The proof can be easily extended to local case, by
replacing ``globally exponentially stable'' with ``exponentially stable with
region of attraction $U$ where $U$ is an invariant set.'' . We start with
exponential stability.

Recall that, in the proof of converse theorems, there is a key
assumption:\ the global Lipschitz condition. In $\mathbb{R}^{n}$,
$f:\mathbb{R}^{n}\rightarrow\mathbb{R}^{n}$ is said to be Lipschitz continuous
if there exists a constant $L$ such that%
\[
|f(x)-f(y)|\leq L|x-y|
\]
where $\left\vert \cdot\right\vert $ is the Euclidean norm. On Riemannian
manifold, if $f$ is a vector field, $f(x)$ and $f(y)$ will live in different
tangent spaces, so it's not possible to compare them directly. In
\cite{taringoo2013local}, the authors considered the tangent map%
\[
Tf:T\mathcal{X}\rightarrow TT\mathcal{X}.
\]
At every point $x\in\mathcal{X}$, $T_{x}f$ is a linear operator. The authors
assume this operator to be uniformly bounded and claim that if%
\begin{equation}
|T_{x}f(t,x)(X)|_{e}\leq c_{2}|T_{x}f(t,x)(X)|_{g}, \label{0}%
\end{equation}
when%
\begin{equation}
|X|_{e}\leq c_{2}|X|_{g},\ \forall X\in T_{x}\mathcal{X} \label{1}%
\end{equation}
where $\left\vert \cdot\right\vert _{e}$ and $\left\vert \cdot\right\vert
_{g}$ stand for the Euclidean and Riemannian metric respectively. However,
$T_{x}f(t,x)(X)$ lives in $T_{f(x)}T\mathcal{X}$ so its Riemannian norm needs
to be defined. There exist canonical Riemannian metrics on the second order
tangent bundle, such as the Sasaki metric, we remark that however, even if
$\left\vert T_{x}f(t,x)(X)\right\vert _{g}$ is replaced by a Riemannian metric
on $TT\mathcal{X}$, the implication from(\ref{1}) to (\ref{0}) is not clear.

Instead of defining a metric on $TT\mathcal{X}$ and studying the tangent map,
we consider the Riemannian version of Lipschitz continuity. This definition
can be found for example in \cite{canary2006fundamentals} Chapter II.3.
Intuitively, we transport two tangent vectors into a same tangent space so
that we can make the comparison between them.

\begin{definition}
[Parallel transport]Let $c:I\rightarrow\mathcal{X}$ be a differentiable curve
in $\mathcal{X}$ and $V_{0}$ a vector tangent to $\mathcal{X}$ at $c(t_{0})$.
Then there exists a unique parallel vector field $V$ along $c$, i.e.
$\nabla_{\gamma^{\prime}(t)}V(t)=0$ such that $V(t_{0})=V_{0}$. We call $V(t)$
the parallel transport of $V(t_{0})$ along $c$.
\end{definition}

For complete Riemannian manifold, given $x,y\in\mathcal{X}$, there exists a
minimizing geodesic curve $\gamma:[0,1]\rightarrow\mathcal{X}$ joining $x$ to
$y$. Given $V\in T_{x}\mathcal{X}$, let $V(t)$ be the parallel transport of
$V$ along $\gamma$, then we denote $P_{x}^{y}V=:V(1)\in T_{y}\mathcal{X}$,
i.e. we transport the vector $V$ in $T_{x}\mathcal{X}$ to $T_{y}\mathcal{X}$.

\begin{definition}
\label{def: Lip}A vector field $X$ on $\mathcal{X}$ is said to be globally
Lipschitz continuous on $\mathcal{X}$, if there exists a constant $L>0$ such
that for $p,q\in\mathcal{X}$ and all $\gamma$ geodesic joining $p$ to $q$,
there holds%
\[
\left\vert P_{p}^{q}X(p)-X(q)\right\vert \leq Ld(p,q)
\]
where $|\cdot|$ is the norm induced by the Riemannian metric.
\end{definition}

It can be easily shown that if $|\nabla_{c^{\prime}(0)}X|\leq L$ for all
$c(t)$ such that $\left\vert c^{\prime}(0)\right\vert =1$, then, $X$ is
Lipschitz continuous with constant $L$. Since the Levi-Civita connection
$\nabla$ is an affine connection, we have $\nabla_{v}V=|v|\nabla_{v/|v|}V$,
consequently, $|\nabla_{c^{\prime}(0)}X|\leq L$ is equivalent to $|\nabla
_{v}V|\leq L|v|$ for all $v$. We will see in the following that in Riemannian
manifold, it's the covariant derivative rather than the tangent map which
takes into play.

\begin{lemma}
\label{lem: esti of Lip}Given that the system is globally Lipschitz continuous
with constant $L$, then there holds the following estimation%
\begin{equation}
d(x_{1},x_{2})e^{-L(\tau-t)}\leq d(\phi(\tau;t,x_{1})),\phi(\tau;t,x_{2}))\leq
d(x_{1},x_{2})e^{L(\tau-t)},\ \forall\tau\geq t,\ \forall x\in\mathcal{X}.
\label{esti of Lip}%
\end{equation}

\end{lemma}

\begin{proof}
Suppose that $x_{1}$ and $x_{2}$ is joined by normalized geodesic
$\gamma:[0,\hat{s}]\rightarrow\mathcal{X}$, with $\gamma(0)=x_{1}$ and
$\gamma(\hat{s})=x_{2}$, where $\hat{s}$ is the length of $\gamma$. Then the
map $F(t,s)=\phi(t;t_{0},\gamma(s))$ defines a variation of $\gamma$. By the
first variation formula of arc length, we have%
\begin{align}
\left.  \frac{d}{d\tau}\right\vert _{\tau=t}d(\phi(\tau;t,x_{1})),\phi
(\tau;t,x_{2})) &  =\left.  \left\langle \frac{\partial\phi(t;t_{0}%
,\gamma(s))}{\partial s},\frac{\partial\phi(t;t_{0},\gamma(s))}{\partial
t}\right\rangle \right\vert _{s=0,\tau=t}^{s=\hat{s},\tau=t}\nonumber\\
&  =\left\langle \gamma^{\prime}(\hat{s}),f(x_{2})\right\rangle -\left\langle
\gamma^{\prime}(0),f(x_{1})\right\rangle \nonumber\\
&  =\left\langle P_{x_{2}}^{x_{1}}\gamma^{\prime}(\hat{s}),P_{x_{2}}^{x_{1}%
}f(x_{2})\right\rangle -\left\langle \gamma^{\prime}(0),f(x_{1})\right\rangle
\nonumber\\
&  =\left\langle \gamma^{\prime}(0),P_{x_{2}}^{x_{1}}f(x_{2})-f(x_{1}%
)\right\rangle ,\label{eq:0}%
\end{align}
where the third equality follows from the inner product preserving property of
the parallel transport operator.
Since $\gamma$ is normalized, by the Lipschitz continuity, we have%
\[
\left\Vert \left.  \frac{d}{d\tau}\right\vert _{\tau=t}d(\phi(\tau
;t,x_{1})),\phi(\tau;t,x_{2}))\right\Vert \leq Ld(x_{1},x_{2}).
\]
Using the semi-group property of the flow, for any $s>t$, we have%
\begin{align*}
&  \left.  \frac{d}{d\tau}\right\vert _{\tau=s}d(\phi(\tau;t,x_{1})),\phi
(\tau;t,x_{2}))\\
&  =\left.  \frac{d}{d\tau}\right\vert _{\tau=s}d(\phi(\tau;s,\phi
(s;t,x_{1})),\phi(\tau;s,\phi(s;t,x_{2}))),
\end{align*}
hence%
\[
\left\Vert \left.  \frac{d}{d\tau}\right\vert _{\tau=s}d(\phi(\tau
;t,x_{1})),\phi(\tau;t,x_{2}))\right\Vert \leq Ld(\phi(s;t,x_{1}%
),\phi(s;t,x_{2})).
\]
Or equivalently,%
\[
-Ld(\phi(\tau;t,x_{1}),\phi(\tau;t,x_{2}))\leq\frac{d}{d\tau}d(\phi
(\tau;t,x_{1})),\phi(\tau;t,x_{2}))\leq Ld(\phi(\tau;t,x_{1}),\phi
(\tau;t,x_{2})),
\]
from which we get (\ref{esti of Lip}).
\end{proof}

\begin{theorem}
\label{converse Exp}Assume that $f(\cdot,x)$ is globally Lipschitz (with
constant $L$). Let $x_{\ast}$ be a UGES equilibrium point of the system
(\ref{Sys}) on the $\mathcal{X}$. Then there exists a Lyapunov candidate
$V(t,x)$ verifying the following properties:

\begin{enumerate}
\item There exist two positive constants $c_{1}$ and $c_{2}$, such that%
\begin{equation}
c_{1}d(x,x_{\ast})\leq V(t,x)\leq c_{2}d(x,x_{\ast}),\ \forall x\in
\mathcal{X}. \label{eq: thm: 1}%
\end{equation}

\item The Lie derivative of $V(t,x)$ in the sense of Definition
\ref{def: tLie} along the system satisfies%
\begin{equation}
\mathcal{L}_{f}V(t,x)\leq-c_{3}V(t,x) \label{eq: thm: 2}%
\end{equation}
where $c_{3}$ is a positive constant.

\item If $d(\cdot,x_{\ast}):\mathcal{X}\rightarrow\mathbb{R}$ is class
$\mathcal{C}^{1}$. Then for every $t$, the differential of $V(t,x)$,
d$V(t,x)\in T^{\ast}\mathcal{X}$ is uniformly bounded on $T^{\ast}\mathcal{X}%
$:%
\begin{equation}
\left\vert \text{d}V(t,x)\right\vert \leq c_{4} \label{eq: thm: 3}%
\end{equation}
where $c_{4}$ is a positive constant independent of $t$ and $x$.
\end{enumerate}
\end{theorem}

\begin{proof}
\textbf{Item 1:}\ Consider the function%
\begin{equation}
V(t,x)=\int_{t}^{t+\delta}d(\phi(\tau;t,x),x_{\ast})\text{d}\tau.
\label{LF candidate}%
\end{equation}
Setting $x_{2}=x_{\ast}$ in Lemma \ref{lem: esti of Lip}, we get the following
estimate by the fact that $x_{\ast}$ is an equilibrium point:%
\begin{equation}
d(x,x_{\ast})e^{-L(\tau-t)}\leq d(\phi(\tau;t,x)),x_{\ast})\leq d(x,x_{\ast
})e^{L(\tau-t)},\ \forall\tau\geq t,\ \forall x\in X. \label{bound of dist}%
\end{equation}
Thus the defined function (\ref{LF candidate}) admits the following bounds:%
\[
V(t,x)=\int_{t}^{t+\delta}d(\phi(\tau;t,x),x_{\ast})\text{d}\tau\geq\int%
_{t}^{t+\delta}d(x,x_{\ast})e^{-L(\tau-t)}\text{d}\tau=\frac{1-e^{-L\delta}%
}{L}d(x,x_{\ast}),
\]
and%
\[
V(t,x)=\int_{t}^{t+\delta}d(\phi(\tau;t,x),x_{\ast})\text{d}\tau\leq\int%
_{t}^{t+\delta}Ke^{-\lambda(\tau-t)}d(x,x_{\ast})\text{d}\tau=\frac
{K(1-e^{-\lambda\delta})}{L}d(x,x_{\ast}).
\]
So we can find two positive constants $c_{1},$ $c_{2}$ such that%
\begin{equation}
c_{1}d(x,x_{\ast})\leq V(t,x)\leq c_{2}d(x,x_{\ast}),\ \forall x\in
\mathcal{X}. \label{bound of LF}%
\end{equation}
\textbf{Item 2:}\ In order to estimate the evolution of $V(t,x)$ along the
system, we again utilize the semi-group property:%
\[
V(s,\phi(s;t,x))=\int_{s}^{s+\delta}d(\phi(\tau;s,\phi(s;t,x)),x_{\ast
})\text{d}\tau=\int_{s}^{s+\delta}d(\phi(\tau;t,x),x_{\ast})\text{d}\tau.
\]
Therefore%
\begin{align}
\mathcal{L}_{f}V(t,x)  &  =\left.  \frac{d}{ds}\right\vert _{s=t}%
V(s,\phi(s;t,x))=d(\phi(t+\delta;t,x),x_{\ast})-d(x,x_{\ast})\nonumber\\
&  \leq-(1-Ke^{-\lambda\delta})d(x,x_{\ast})\nonumber\\
&  =-K^{\prime}d(x,x_{\ast}),\ \forall t,\ \forall x\in\mathcal{X},
\label{Neg of dV}%
\end{align}
where $\delta$ is chosen such that $K^{\prime}>0$. By (\ref{bound of LF}),
\[
\mathcal{L}_{f}V(t,x)\leq-K^{\prime}d(x,x_{\ast})\leq-\frac{K^{\prime}}{c_{2}%
}V(t,x).
\]
Now $c_{3}$ can be set as $c_{3}=K^{\prime}/c_{2}$.
\textbf{Item 3: }Denote $h_{t}(x)=V(t,x)$, then for any $v\in T_{x}%
\mathcal{X}$,
\[
\text{d}h_{t}(v)=\left.  \frac{\text{d}}{\text{d}s}\right\vert _{s=0}%
h_{t}(c(s))
\]
where $c:[-\varepsilon,\varepsilon]\rightarrow\mathcal{X}$ is a smooth curve
with $c^{\prime}(0)=v$. So%
\begin{equation}
\text{d}h_{t}(v)=\int_{t}^{t+\delta}\left.  \frac{\text{d}}{\text{d}%
s}\right\vert _{s=0}d(\phi(\tau;t,c(s)),x_{\ast})\text{d}\tau\label{pf: 0}%
\end{equation}
By first variation formula,
\[
\left.  \frac{\text{d}}{\text{d}s}\right\vert _{s=0}d(\phi(\tau
;t,c(s)),x_{\ast})=\left\langle \phi(\tau;t,x)_{\ast}v,\gamma^{\prime
}(1)\right\rangle ,
\]
where $\phi(\tau;t,x)_{\ast}v$ is the push forward of the vector $c^{\prime
}(0)$ by the map $x\mapsto\phi(\tau;t,x)$ and $\gamma$ is the normalized
geodesic joining $x_{\ast}$ to $\phi(\tau;t,x)$. Hence
\begin{equation}
\left\Vert \left.  \frac{\text{d}}{\text{d}s}\right\vert _{s=0}d(\phi
(\tau;t,c(s)),x_{\ast})\right\Vert ^{2}\leq\left\langle \phi(\tau;t,x)_{\ast
}v,\phi(\tau;t,x)_{\ast}v\right\rangle . \label{pf: 1}%
\end{equation}
Now we estimate the term on the right hand side.
\begin{align*}
\frac{\text{d}}{\text{d}\tau}\frac{1}{2}\left\langle \phi(\tau;t,x)_{\ast
}v,\phi(\tau;t,x)_{\ast}v\right\rangle  &  =\left\langle \frac{\text{D}%
}{\text{d}\tau}\phi(\tau;t,x)_{\ast}v,\phi(\tau;t,x)_{\ast}v\right\rangle \\
&  =\left\langle \nabla_{f(\phi(\tau;t,x))}\phi(\tau;t,x)_{\ast}v,\phi
(\tau;t,x)_{\ast}v\right\rangle \\
&  =\left\langle \nabla_{\phi(\tau;t,x)_{\ast}v}f(\phi(\tau;t,x)),\phi
(\tau;t,x)_{\ast}v\right\rangle \\
&  +\left\langle [f(\phi(\tau;t,x)),\phi(\tau;t,x)_{\ast}v],\phi
(\tau;t,x)_{\ast}v\right\rangle ,
\end{align*}
where we have used the symmetry of the Levi-Civita connection, i.e.%
\[
\nabla_{X}Y-\nabla_{Y}X=[X,Y].
\]
However,%
\begin{align*}
\lbrack f(\phi(\tau;t,x)),\phi(\tau;t,x)_{\ast}v]  &  =L_{f(\phi(\tau
;t,x))}\phi(\tau;t,x)_{\ast}v\\
&  =\frac{\text{d}}{\text{d}\tau}\phi^{\ast}\phi_{\ast}v=\frac{\text{d}%
}{\text{d}\tau}v=0.
\end{align*}
Therefore%
\begin{align*}
\frac{\text{d}}{\text{d}\tau}\frac{1}{2}\left\langle \phi(\tau;t,x)_{\ast
}v,\phi(\tau;t,x)_{\ast}v\right\rangle  &  =\left\langle \nabla_{\phi
(\tau;t,x)_{\ast}v}f(\phi(\tau;t,x)),\phi(\tau;t,x)_{\ast}v\right\rangle \\
&  \leq L\left\langle \phi(\tau;t,x)_{\ast}v,\phi(\tau;t,x)_{\ast
}v\right\rangle
\end{align*}
where we have used that fact that $\left\vert \nabla_{v}f\right\vert \leq
L|v|$. So%
\[
\left\langle \phi(\tau;t,x)_{\ast}v,\phi(\tau;t,x)_{\ast}v\right\rangle \leq
e^{2L(\tau-t)}|v|^{2}%
\]
or%
\[
\left\vert \phi(\tau;t,x)_{\ast}v\right\vert \leq e^{L(\tau-t)}|v|,\ \forall
\tau\geq t.
\]%
\[
\text{d}h_{t}(v)\leq\int_{t}^{t+\delta}e^{L(\tau-t)}|v|\text{d}\tau
=\frac{e^{L\delta}-1}{L}|v|.
\]
So we have obtained%
\[
\left\vert \text{d}_{x}V(t,x)(v)\right\vert \leq c_{4}|v|
\]
or%
\[
|\text{d}_{x}V(t,x)|\leq c_{4},\ \forall x\in\mathcal{X},
\]
where $c_{4}=(e^{L\gamma}-1)/L$.
\end{proof}

\begin{remark}
\label{rmk: item 3}In Item 3, we have asked $d(x,x_{\ast})$ to be
differentiable. This is however, not guaranteed in general. For example, in
Euclidean space, $d(x,x_{\ast})=|x-x_{\ast}|$, which is not differentiable at
the point $x_{\ast}$. A possible way to resolve this problem is to consider
the following Lyapunov candidate:
\begin{equation}
V(t,x)=\int_{t}^{t+\gamma}d(\phi(\tau;t,x),x_{\ast})^{p}\text{d}%
\tau\label{LF candidate p}%
\end{equation}
where $p\geq1$. In the Euclidean case, we set $p=2$, and $d(x,x_{\ast}%
)^{2}=|x-x_{\ast}|^{2}$ is differentiable. In effect, it can be easily
verified (\ref{LF candidate p}) can still serve as a Lyapunov function. The
proof can be carried out in exactly the same way as Theorem \ref{converse Exp}%
. However, when $p\neq1$, the claims of Item 1 and 3 should change
accordingly. For example, the estimate of $V(t,x)$ becomes
\[
c_{1} d(x,x_{*})^{p} \le V(t,x) \le c_{2} d(x,x_{*})^{p}
\]
and $dV$ satisfies
\[
\left\vert \text{d}V(t,x)\right\vert \leq c_{4}d(x,x_{\ast})^{p-1}.
\]

\end{remark}

\begin{remark}
In contrast to the proof in \cite{taringoo2013local}, all the proof here is
coordinate-free. So if the system has an invariant set $U$ as region of
attraction, then a Lyapunov function can be naturally defined everywhere on
$U$. Moreover, the Lyapunov function can be constructed such that $V(t,x)
\rightarrow\infty$ when $x$ approaches the boundary of $U$.
\end{remark}

\begin{remark}
In Euclidean space, the Lyapunov candidate becomes%
\[
V(t,x)=\int_{t}^{t+\gamma}|\phi(\tau;t,x)|^{2}\text{d}x
\]
by setting $p=2$ in (\ref{LF candidate p}), reducing to the standard
construction, see \cite{khalil2002nonlinear}.
\end{remark}

\section{Discussions and applications\label{sec: dis and ext}}

Theorem \ref{converse Exp} requires the system to be globally stable. However,
such requirement is not essential. In fact, all the procedures of the proof
can be done locally in the same manner. Hence we can obtain local version of
converse theorems.

The extension to asymptotically stability is also not difficult. Following
\cite{khalil2002nonlinear}, we just need to modify the Lyapunov candidate to%
\[
V(t,x)=\int_{t}^{\infty}G(d(\phi(\tau;t,x),x_{\ast}))\text{d}\tau
\]
where $G$ is constructed from the following Massera's lemma.

\begin{lemma}
[Massera]Let $g:\mathbb{R}_{+}\rightarrow\mathbb{R}$ be a positive,
continuous, strictly decreasing function with $g(t)\rightarrow0$ as
$t\rightarrow\infty$. Let $h:\mathbb{R}_{+}\rightarrow\mathbb{R}$ be a
positive, continuous, non decreasing function. Then, there exists a function
$G(t)$ such that
\end{lemma}

\begin{enumerate}
\item $G$ and its derivative $G^{\prime}$ are class $\mathcal{K}$ functions
defined for all $t\geq0$;

\item For any continuous function $u(t)$ that satisfies $0\leq u(t)\leq g(t)$
for all $t\geq0$, there exist positive constants $k_{1}$ and $k_{2}$,
independent of $u$, such that%
\[
\int_{0}^{\infty}G(u(t))\text{d}t\leq k_{1};\ \int_{0}^{\infty}G^{\prime
}(u(t))h(t)\text{d}t\leq k_{2}.
\]

\end{enumerate}

The rest of the proof can be done similarly as that of Theorem
\ref{converse Exp}. Hence we have the following theorem.

\begin{theorem}
\label{converse Asmp}Assume that $f(\cdot,x)$ is globally Lipschitz (with
constant $L$) in the sense of \ref{def: Lip}. Let $x_{\ast}$ be a UGAS
equilibrium point of the system (\ref{Sys}) on the $\mathcal{X}$, i.e.%
\[
d(\phi(t;t_{0},x_{0}),x_{\ast})\leq\beta(d(x_{0},x_{\ast}),t-t_{0}),\ \forall
t\geq t_{0},\ x_{0}\in\mathcal{X}%
\]
for a class $\mathcal{KL}$ function $\beta$. Then there exists a Lyapunov
candidate $V$ verifying the following two properties:

\begin{enumerate}
\item There exist a $\mathcal{C}^{1}$ function $V$, such that%
\[
\alpha_{1}(d(x,x_{\ast}))\leq V(t,x)\leq\alpha_{2}(d(x,x_{\ast})),\ \forall
x\in\mathcal{X}.
\]

\item The Lie derivative of $V(t,x)$ along the system satisfies%
\[
L_{f}V(t,x)\leq-\alpha_{3}(V(t,x))
\]

\item If $d(\cdot,x_{\ast}):\mathcal{X}\rightarrow\mathbb{R}$ is class
$\mathcal{C}^{1}$. Then for every $t$, the differential of $V(t,x)$,
d$V(t,x)\in T^{\ast}\mathcal{X}$ is uniformly bounded on $T^{\ast}\mathcal{X}%
$:%
\begin{equation}
\left\vert \text{d}V(t,x)\right\vert \leq\alpha_{4}(V(t,x))
\end{equation}
where $\alpha_{i},\ i=1,2,3,4$ are class $\mathcal{K}$ functions.
\end{enumerate}
\end{theorem}

As application, we show that Theorem \ref{converse Exp} can be applied to
prove the input-to-state stability (ISS)\ of a class of systems. The classical
form of this theorem can be found in \cite{khalil2002nonlinear}.

\begin{corollary}
Consider the control system%
\begin{equation}
\dot{x}=f(t,x,u) \label{NL sys input}%
\end{equation}
on Riemannian manifold $\mathcal{X}$, where $f$ is $\mathcal{C}^{1}$ and
globally Lipschitz in $x$. Additionally, we assume $f$ is globally Lipschitz
in $u$ with constant $L$, i.e.%
\[
|f(t,x,u)-f(t,x,0)|\leq L|u|.
\]

If the unforced system $\dot{x}=f(t,x,0)$ is UGES with respect to equilibrium
point $x=0$, then the system (\ref{NL sys input}) is ISS.
\end{corollary}

\begin{proof}
By Theorem (\ref{converse Exp}), a Lyapunov function $V(t,x)$ verifying the
three conditions can be constructed for the unforced system $\dot{x}=f(t,x,0)$
when $\delta$ is large enough. Rewrite%
\[
f(t,x,u)=f_{1}+f_{2}%
\]
where%
\begin{align*}
f_{1}  &  =f(t,x,0)\\
f_{2}  &  =f(t,x,u)-f(t,x,0).
\end{align*}
By assumption, $\mathcal{L}_{f_{1}}V\leq-c_{3}V$. The Lie derivative of
$V(t,x)$ with respect to (\ref{NL sys input}) reads%
\begin{align*}
\mathcal{L}_{f}V  &  =\mathcal{L}_{f}V=\frac{\partial V}{\partial t}%
+L_{f}V=\frac{\partial V}{\partial t}+\text{d}V\left(  f_{1}+f_{2}\right) \\
&  =\mathcal{L}_{f_{1}}V+\text{d}V\left[  f(t,x,u)-f(t,x,0)\right] \\
&  \leq-c_{3}V+c_{4}|f(t,x,u)-f(t,x,0)|\\
&  \leq-c_{3}V+c_{4}L|u|_{\infty}.
\end{align*}
Now invoking standard arguments from ISS theory, we conclude that the system
(\ref{NL sys input}) is ISS.
\end{proof}

\section{Concluding remarks}

We have proved the converse theorems on Riemannian manifolds, in a
coordinate-free way. The constructed Lyapunov functions and the line of proofs
share a lot in common with that in Euclidean space. This may suggests that
there is no essential difference of Lyapunov stabilities between Riemannian
manifolds and Euclidean space in regardless of the global topology. Further
studies may include the application of the results and the extension to
Finsler manifolds.

\section{Acknowledgement}

We thank Dr. Antoine Chaillet for his instructions and fruitful discussions
during the preparation of this manuscript.

\bibliographystyle{elsarticle-num}
\bibliography{C:/swp55/TCITeX/BibTeX/bib/USEROWN/GeometricControl}

\end{document}